\documentclass[12pt,reqno,centertags]{amsart}
\usepackage{a4}
\usepackage{amsmath,amsthm,amscd,amssymb}
\usepackage{latexsym}
\usepackage{upref}
\usepackage{hyperref}
\usepackage{upgreek}
\usepackage{textgreek}
\usepackage{cite}
\usepackage{tikz}
\usepackage{enumerate}
\usepackage{multicol}
\usepackage{accents}
\usepackage{lmodern}
\usepackage{microtype}
\usepackage{graphicx}
\usepackage{subcaption}
\usepackage{dsfont}
\usepackage{colonequals}
\usepackage{bbm}
\usepackage{fullpage}
\UseRawInputEncoding

\newcommand{\abs}[1]{\left|#1\right|}

\newcommand\lrb[1]{\left\lbrace#1 \right\rbrace}
\newcommand\lrp[1]{\left(#1 \right)}
\newcommand\cA[1]{{\mathcal #1}}

\newcommand\C{{\mathbb C}}
\newcommand\N{{\mathbb N}}
\newcommand\R{{\mathbb R}}
\newcommand\Z{{\mathbb Z}}

\newcommand\cc[1]{\overline {#1}}
\newcommand\Dn[2]{\left |{#1}\right\rangle\!\!\left\langle {#2} \right |}

\newcommand\vk[1]{\left |{#1}\right\rangle\!}
\newcommand\ip[2]{\left\langle {#1},{#2} \right\rangle}
\newcommand\no[1]{\left\| {#1} \right\|}
\newcommand\unit{\hbox{\rm 1\kern-2.8truept l}}
\newcommand\tr[1]{{{\rm tr}}\left(#1\right)}

\DeclareMathOperator{\supp}{supp}

\DeclareMathOperator{\ran}{ran}
\DeclareMathOperator{\Span}{span}

\allowdisplaybreaks
\numberwithin{equation}{section}
\newtheorem{theorem}{Theorem}[section]
\newtheorem{lemma}[theorem]{Lemma}
\newtheorem{corollary}[theorem]{Corollary}

\newtheorem{proposition}[theorem]{Proposition}
\theoremstyle{definition}
\newtheorem{definition}[theorem]{Definition}
\newtheorem{example}[theorem]{Example}

\theoremstyle{remark}
\newtheorem{remark}[theorem]{Remark}

\begin{document}
\title[Quantum error-correcting codes via inner products and error bases]{Quantum error-correcting codes via inner products and error bases }

\author[Jorge R. Bola\~nos-Serv\'in, Yuriko Pitones and Josu\'e I. Rios-Cangas]{Jorge R. Bola\~nos-Serv\'in, Yuriko Pitones and Josu\'e I. Rios-Cangas}
\address{Departamento de Matem\'aticas, Universidad Aut\'onoma Metropolitana, Unidad Iztapalapa,  San Rafael Atlixco 186, 093110 Iztapalapa, Ciudad de M\'exico.}
\email{jrbs@xanum.uam.mx,\quad ypitones@xanum.uam.mx,\quad jottsmok@xanum.uam.mx}
\date{\today}

\begin{abstract}
In this paper, we address the problem of state communication in finite-level quantum systems through noise-affected channels. Our approach is based on a self-consistent theory of decoding inner products associated with the code and error (or noise) bases defined on corrupting subspaces. This viewpoint yields new necessary and sufficient conditions for the existence of quantum error-correcting codes in terms of these inner products. The obtained results extend the foundations of quantum error correction beyond classical analogies, highlighting the structural insights offered by operator theory and the underlying product space.
\end{abstract}

\maketitle

\section{Introduction}  

It is well-known that the quantum error-correction theory addresses the problem of protecting quantum information from errors due to decoherence, commonly referred to as noise. Quantum error correction is essential to achieving a fault-tolerant quantum information theory that can reduce the effects of noise in stored quantum information, faulty quantum gates, faulty preparation of quantum states, and faulty measurements.  This theory has several applications; in particular, it enables quantum computers with low qubit fidelity to run algorithms of greater complexity or greater circuit depth \cite{MR2181442,CAI202150}.

Although quantum error correction is related to classical error correction, classical techniques cannot be directly applied to quantum systems, since classical measurements cannot preserve phase information in entangled states. In addition, the information cloning method to restore corrupted information in classical communication \cite{MR465509,MR465510} fails within the framework of quantum information theory \cite{MR2193880}. Nevertheless, Shor \cite{PhysRevA.52.R2493} shows that it is possible to correct a quantum state affected by a specific known noise by spreading the information over many states supported on an encoding. It is worth pointing out that Shor \cite{PhysRevA.52.R2493} and Steane \cite{PhysRevLett.77.793} are considered pioneers of quantum error-correction codes theory to introduce the first codes that correct single-bit flip and single-phase flip errors. Since then, several codes have been implemented that correct for specific interactions \cite{MR1421749,PhysRevA.54.1098,PhysRevLett.77.198,PhysRevA.54.3824,PhysRevA.54.R1745}, which have set the stage for general quantum error correction theory to implement practical self-consistent quantum theories and achieve highly reliable quantum communication, for instance, the Knill-Laflamme criteria \cite{MR1455854} for decoding quantum error correction codes. To the best of our knowledge, quantum error correction schemes have predominantly been studied over finite fields, likely motivated by practical considerations and their connection to classical error correction theory, though as we demonstrate here, the general case of Hilbert spaces over $\mathbb C$ gives rise to a sufficiently rich and interesting theory in its own right which has been left aside.
 
 In this paper, we are interested in addressing the problem of communicating and correcting the states of a $n$-level quantum system through a channel disturbed by noise in the setting of a finite dimensional Hilbert space over the complex field $\mathbb C$ with inner product $\ip{\cdot}{\cdot}$ that is anti-linear in the first argument. The von Neumann algebra of all bounded linear operators on $\cA H$ is denoted by $\mathcal {B(H)}$, which is naturally endowed with a Hilbert space structure, namely, $\mathcal{(B(H)},\ip{\cdot}{\cdot}_2)$, where  
\begin{gather}\label{eq:inner-product-S2}
\ip{\eta}{\tau}_2=\tr{\eta^*\tau}\,,\qquad \eta_\tau\in \mathcal{B(H)}\,,
\end{gather} and of particular relevance are the states (or density operators) which are a positive selfadjoint operators with unit trace.
A \emph{quantum code} (or simply code) is any non-trivial subspace $\cA C$  of $\cA H$, while a \emph{noise set} (or error set) is any non-trivial subspace  $\cA N$ of $\mathcal {B(H)}$.

The basic communication scheme is, an input state $\rho_{\rm in}$ supported in $\cA C$, is transmitted through a  channel and received as an output state $\rho_{\rm out}$ of the form (cf. \cite{MR3051751})
\begin{gather*}
\rho_{\rm out}=\frac{1}{\tr{\rho N^*N}}N\rho N^*\,,\quad N\in\cA N\,.
\end{gather*}
The well-known Knill-Laflamme criteria characterize the existence of  a decoding quantum channel $\Phi\colon\cA B(\cA H)\to \cA B(\cA H)$, namely, a channel which recovers $\rho_{\rm in}$ from $\rho_{\rm out}$, i.e.,
\begin{gather}\label{eq:N-cc-aux}
 \Phi(\rho_{\rm out})=\rho_{\rm in}; \ \ \ \ \ \Phi\lrp{N\rho_{\rm in} N^*}=\tr{\rho_{\rm in} N^*N}\rho_{\rm in}\,.
\end{gather}

The aim of this paper is to present a connection between the existence of an error-correction code $\cA C$ and  suitable inner products on the error subspace  $\cA N$ in the case of a complex Hilbert space. To achieve this, Section~\ref{s:TCDIP} introduces the notion of a $\cA C$-decoding inner product, which is an inner product on the subspace $\cA N$. This inner product can be expressed in terms of the Hilbert-Schmidt inner product (see Theorem~\ref{th:repres-Cdip}). Although $\cA N$ can always be equipped with the standard inner product, it is not always possible to endow it with a $\cA C$-decoding inner product (Remark~\ref{rm:necessity-dimN}). In Section~\ref{s:QECC-BIP}, we develop the theory of quantum error-correcting codes using the code-decoding inner product in the presence of non-negligible noise (see Remark~\ref{rmk:nonzero-cond}). We prove that the existence of an $\cA N$-correcting code $(\cA C,\Phi)$ is equivalent to the existence of the $\cA C$-decoding inner product on $\cA N$ (see Theorem~\ref{th:KL-G}). Furthermore, Theorem~\ref{th:N-cc-Conditions} provides  necessary and sufficient conditions for the existence of $\cA N$-correcting code, and introduce the concept of a code-decoding noise basis of $\cA N$ (see Definition~\ref{def:C-dnb-N}), which explicitly generates quantum channels that decode $\cA C$. A necessary condition for the existence of an $\cA N$-correcting code is given by $\dim \cA N=\dim\cA H/\dim \cA C$ by means of  Remark~\ref{rm:necessity-dimN} and Theorem~\ref{th:N-cc-Conditions}. In order to describe the elements of code-decoding noise bases, Section~\ref{sec:CDPINB} is divided into three parts. Subsection~\ref{eq:vNW-Deco} focuses on the von Neumann-Wold decomposition, which decomposes any partial isometry into its unitary and completely non-unitary components (Theorem~\ref{th:characterization-PIO}). Subsection~\ref{ss:Code-bypowers-isom} characterizes code-decoding noise bases generated by linear combinations of powers of partial isometries (Theorem~\ref{th:codes-Vjs}). This section also discusses the existence of the largest error-correcting code and presents an example of an error-correcting code that lacks the largest error-correcting code (Example~\ref{ex:ne-lcc}). In Subsection~\ref{s:shift}, we present code-decoding noise bases using shift and clock operators (see Theorems~\ref{th:eio-codes} and~\ref{th:cio-codes}). These operators play a central role in the theory of stabilizer codes and t-error correcting quantum codes, as they define the Weyl operators (see Remark~\ref{rm:weyl-operators}). The last section is dedicated to code-decoding noise bases in reducing subspaces, Corollary~\ref{cor:wNw-Deco-RS} and Theorems~\ref{th:rS-Pisom-Ncc} and~\ref{th:codes-Rs-ONh} generalize the principal results of Section~\ref{sec:CDPINB}.

\section{The code-decoding inner product}\label{s:TCDIP} We begin by introducing some definitions and the notion of an inner product on $\cA N$. Recall that the support of an operator $T\in\cA B(\cA H)$ is the subspace
\begin{gather*}
\supp T\colonequals (\ker T)^\perp=\ran T^*\,.
\end{gather*}
For simplicity of notation, o.n.b. denotes \emph{orthonormal basis}, $\delta_{rj}$ denotes the \emph{Kronecker delta function} and we say $\cA F\leq \cA H$ whenever $\cA F$ is a subspace of $\cA H$ and $P_{\cA F}\in\mathcal {B(H)}$ represents the orthogonal projection onto $\cA F\leq \cA H$. Finally, we consider the convex subset of states supported in $\cA C$
\begin{gather*}
\tilde{\cA C}\colonequals\lrb{\rho\in\cA B(\cA H)\,:\, \supp \rho\leq\cA C} \subset \mathcal{B}(\cA H)\,.
\end{gather*} 
\begin{definition} An inner product $\varphi\colon\cA N\times \cA N\to \C$ that satisfies 
\begin{gather}\label{eq:C-cc-Condition}
\eta N^*M\rho=\varphi(N,M)\eta\rho\,,\quad\mbox{for all}\quad \eta,\rho\in\tilde{
\cA C}\mbox{ and }N,M\in\cA N\,,
\end{gather}  
 is called the \emph{$\cA C$-decoding} inner product on $\cA N$. It readily follows that the existence of the $\cA C$-decoding inner product is unique on $\cA N$.
\end{definition}

\begin{remark}\label{rm:partial-isometry}
An operator $V\in\cA B(\cA H)$ is \emph{partial isometry} if $V^*V$ is the orthogonal projection onto $\supp V$. When $\supp V=\cA F \leq\cA H$, we say that the partial isometry $V$  is onto.
\end{remark}

\begin{theorem}\label{teo:prop-ipN}
If $\varphi$ is an inner product on $\cA N$, then for all $N,M\in\cA N$, $\eta,\rho\in\tilde{\cA C}$,  $u,v\in\cA C$, and $P_{\cA F}$ the orthogonal projection onto $\cA F\leq\cA C$, the following conditions are equivalent:
\begin{enumerate}[(i)]
\item\label{c1-ipp} $\varphi$ is the $\cA C$-decoding inner product.
\item\label{c2-ipp} $\ip{Nu}{Mv}=\varphi(N,M)\ip uv$. 
\item\label{c3-ipp} $\varphi(N,M)=\tr{\rho N^*M}$, and $\ip{Nu}{Mv}=0$ if $u\perp v$.
\item\label{c4-ipp} $P_{\cA F}N^*MP_{\cA F}=\varphi(N,M)P_{\cA F}$.
\item\label{c5-ipp} $\varphi(N,N)^{-1/2}NP_{\cA F}$ is partial isometry onto $\cA F$, for $N\neq0$.
\end{enumerate}
\end{theorem}
\begin{proof}
\eqref{c1-ipp}$\Rightarrow$\eqref{c2-ipp}: We may assume $u,v\neq0$, since otherwise is straightforward. Thus, by \eqref{eq:C-cc-Condition},
\begin{align*}
\ip{Nu}{Mv}=\frac{1}{\no{u}^2\no v^2}\ip{u}{\Dn uu N^*M\Dn vvv}=\varphi(N,M)\ip{u}{v}\,.
\end{align*}
\eqref{c2-ipp}$\Rightarrow$\eqref{c3-ipp}: It is clear that $\ip{Nu}{Mv}=0$, if $u\perp v$. Besides, for $\rho\in\tilde{\cA C}$ there exists an o.n.b. $\{u_r\}_{r=1}^s$ of $\cA C$ such that $\rho=\sum_{r=1}^s\rho_r\Dn{u_r}{u_r}$, with $\rho_r\geq0$ and $\sum_{r=1}^s\rho_r=1$. Therefore, $\tr{\rho N^*N}=\sum_{r=1}^s\rho_r\ip{Nu_r}{Nu_r}=\varphi(N,N)
$, and by polarization identity,
\begin{align*}
\varphi(N,M)&=\frac14\sum_{k=0}^3i^k\varphi\lrp{i^kN+M,i^kN+M}\\&=\frac14\sum_{k=0}^3i^k\tr{\rho\lrp{i^kN+M}^*\lrp{i^kN+M}}=\tr{\rho N^*M}\,.
\end{align*}
\eqref{c3-ipp}$\Rightarrow$\eqref{c4-ipp}: It follows that  $P_{\cA F}=\sum_{r=1}^l\Dn{u_r}{u_r}$, for some o.n.b. $\{u_r\}_{r=1}^l$ of $\cA F$, and 
\begin{gather*}
P_{\cA F}N^*MP_{\cA F}=\sum_{r=1}^l\sum_{t=1}^l\Dn{u_r}{u_r}N^*M\Dn{u_t}{u_t}=\sum_{r=1}^l\tr{\Dn{u_r}{u_r}N^*M}\Dn{u_r}{u_r}=\varphi(N,M)P_{\cA F}\,.
\end{gather*}

\eqref{c4-ipp}$\Rightarrow$\eqref{c5-ipp}: One simply computes that 
\begin{gather*}
\lrp{\frac{1}{\sqrt{\varphi(N,N)}}NP_{\cA F}}^*\lrp{\frac{1}{\sqrt{\varphi(N,N)}}NP_{\cA F}}=\frac{1}{\varphi(N,N)}P_{\cA F}N^*NP_{\cA F}=P_{\cA F}\,.
\end{gather*}

\eqref{c5-ipp}$\Rightarrow$\eqref{c1-ipp}: Since $\eta,\rho\in\tilde{\cA C}$, one has that  $\eta P_{\cA C}=\eta$, $P_{\cA C}\rho=\rho$, and for $N\neq 0$, 
\begin{gather*}
\frac{1}{\varphi(N,N)}\eta N^*N\rho=\eta \lrp{\frac{1}{\sqrt{\varphi(N,N)}}NP_{\cA C}}^*\lrp{\frac{1}{\sqrt{\varphi(N,N)}}NP_{\cA C}}\rho=\eta\rho\,,
\end{gather*}
i.e., $\eta N^*N\rho=\varphi(N,N)\eta\rho$, for all $N\in\cA N$. Hence, the polarization identity yields
\begin{align*}
\eta N^*M\rho&=\frac14\sum_{k=0}^3i^k\eta (i^kN+M)^*(i^kN+M)\rho\\
&=\frac14\sum_{k=0}^3i^k\varphi(i^kN+M,i^kN+M)\eta\rho=\varphi(N,M)\eta\rho\,,
\end{align*}
i.e., $\varphi$ is the $\cA C$-decoding inner product on $\cA N$.
\end{proof}

Observe in Theorem~\ref{teo:prop-ipN}.\eqref{c2-ipp} that the number $\ip{Nu}{Mu}$ does not depend on the choice of $u\in\cA C$, when $\no u=1$.

\begin{corollary}\label{cor:dim-prop}
If $\varphi$ is the $\cA C$-decoding inner product on $\cA N$ then the following conditions are true: 
\begin{enumerate}[(i)]
\item\label{c1:CcciP} $\cA C\leq\supp  N$,  for all non-zero $N\in\cA N$.
\item\label{c2:CcciP}  For an o.n.b. $\{N_j\}_{j=1}^r$ of $(\cA N,\varphi)$, it follows that every $N_jP_{\cA C}$ is partial isometry onto $\cA C$ and 
\begin{gather}\label{eq:orthogonal-Ccc}
\ran N_jP_{\cA C}\perp \ran N_kP_{\cA C}\,,\quad j\neq k\,.
\end{gather}
\item\label{c3:CcciP} $\dim \cA N\cdot\dim \cA C\leq \dim \cA H$.
\end{enumerate}
\end{corollary}
\begin{proof}
\eqref{c1:CcciP}: For non-zero $u\in\cA C$, one has by Theorem~\ref{teo:prop-ipN}.\eqref{c2-ipp} that $\no{Nu}^2=\varphi(N,N)\no u^2>0$, i.e., $Nu\neq0$, which implies that $\cA C\leq\supp  N$.

\eqref{c2:CcciP}: It is clear from Theorem~\ref{teo:prop-ipN}.\eqref{c5-ipp} that $N_jP_{\cA C}$ is partial isometry onto $\cA C$. Besides, for $j\neq k$, it follows by Theorem~\ref{teo:prop-ipN}.\eqref{c4-ipp} that
\begin{gather*}
(N_jP_{\cA C})^*N_kP_{\cA C}=P_{\cA C}N_j^*N_kP_{\cA C}=\varphi(N_j,N_k)P_{\cA C}=0\,,
\end{gather*}
which implies that $\ran N_kP_{\cA C}\leq \ker(N_jP_{\cA C})^*=(\ran N_jP_{\cA C})^\perp$, i.e., \eqref{eq:orthogonal-Ccc}.

\eqref{c3:CcciP}: Consider an o.n.b. $\{N_j\}_{j=1}^r$ of $(\cA N,\varphi)$, which by item~\eqref{c3:CcciP} every $N_jP_{\cA C}$ is partial isometry onto $\cA C$, viz. $\dim\ran N_jP_{\cA C}=\dim \cA C$. Besides, \eqref{eq:orthogonal-Ccc} holds $\oplus_{j=1}^r\ran N_jP_{\cA C}\leq\cA H$ and
\begin{gather*}
\dim \cA H\geq \dim \bigoplus_{j=1}^r\ran N_jP_{\cA C}=\sum_{j=1}^r\dim\ran N_jP_{\cA C}=r\dim\cA C\,,
\end{gather*}
as required.
\end{proof}

\begin{remark}\label{rm:necessity-dimN}
It is worth pointing out from Corollary~\ref{cor:dim-prop}.\eqref{c3:CcciP} that  if $\dim \cA H<\dim \cA N\cdot\dim \cA C$ then the noise subspace $\cA N$ cannot be endowed by the $\cA C$-decoding inner product.
\end{remark}
Inner products are closely related with positive invertible operators on Hilbert spaces. Namely \cite[Sect. 10.1]{MR1192782}, every inner product $\varphi$ on a Hilbert space $(\cA H,\ip{\cdot}{\cdot})$ admits the representation 
\begin{gather}\label{rmk:inner-PD}
\varphi (f,g)=\ip{f}{Tg}\,,\qquad\mbox{where $T\in\cA B(\cA H)$, is positive and invertible}.
\end{gather}

\begin{lemma}\label{lem:ip-VpN}
For a sesquilinear form $\varphi\colon\cA H\times\cA H\to \C$, it follows that $\varphi$ is an inner product on $\cA H$ if and only if there exists an o.n.b. $\{u_j\}_{j=1}^{t}$ of $(\cA H,\ip\cdot\cdot)$ such that 
\begin{gather}\label{eq:ip-vP-H2}
\varphi(f,g)=\sum_{j=1}^{t}\no{u_j}_{\varphi}^2\ip{f}{u_j}\ip{u_j}{g}\,,\quad\mbox{with}\quad \no{u_j}_\varphi^2=\varphi(u_j,u_j)>0\,.
\end{gather}
In such a case, $\{\no{u_j}_\varphi^{-1} u_j\}_{j=1}^t$ is an o.n.b. of $(\cA H,\varphi)$.
\end{lemma}
\begin{proof}
If $\varphi$ is an inner product on $\cA H$, then by virtue of \eqref{rmk:inner-PD}, there exists a unique positive invertible operator $T\in\cA B(\cA H)$ such that 
$\varphi(f,g)=\ip{f}{Tg}$. In this fashion, there is an o.n.b. $\{u_j\}_{j=1}^{t}$ of $(\cA H,\ip\cdot\cdot)$ such that $T=\sum_{j=1}^{t}p_j\Dn{N_j}{N_j}$, with $p_j>0$, which implies 
\begin{gather}\label{eq:ip-poP}
\varphi(f,g)=\sum_{j=1}^{n-1}p_j\ip{f}{u_j}\ip{u_j}{g}\,,\quad \mbox{whence}\quad \varphi(u_j,u_j)=p_j\,.
\end{gather}
Conversely, $T=\sum_{j=1}^{t}\no{u_j}_\varphi\Dn{u_j}{u_j}$ is a positive invertible operator in $\cA B(\cA H)$ and by \eqref{eq:ip-vP-H2}, $\varphi(f,g)=\ip{f}{Tg}$. Hence, it follows by \eqref{rmk:inner-PD} that $\varphi$ defines an inner product on $\cA H$. Note that $\varphi(\no{u_j}_\varphi^{-1} u_j,\no{u_k}_\varphi^{-1} u_k)=\delta_{jk}$, i.e., $\{\no{u_j}_\varphi^{-1} u_j\}_{j=1}^t$ is an o.n.b. of $(\cA H,\varphi)$. 
\end{proof}

We conclude the section by presenting another characterization of the $\cA C$-decoding inner product in terms of the Hilbert-Schmidt inner product \eqref{eq:inner-product-S2}.

\begin{theorem}\label{th:repres-Cdip}
Let $\varphi$ be a sesquilinear form on $\cA N$. Then $\varphi$ is the $\cA C$-decoding inner product  if and only if there exists an o.n.b. $\{N_j\}_{j=1}^{t}$ of $(\cA N,\ip\cdot\cdot_2)$ such that 
 \begin{gather}\label{eq:ip-vP-L2}
 \varphi(N,M)=\sum_{j=1}^{t}\no{N_j}_{\varphi}^2\ip{N}{N_j}_2\ip{N_j}{M}_2\,,\quad \mbox{where}\quad \no{N_j}_\varphi^2=\varphi(N_j,N_j)>0\,,\\\label{eq:ip-vP-L3}
 \mbox{and}\quad {N_j}_{\cA C}=\no{N_j}_\varphi^{-1} N_jP_{\cA C}\quad \mbox{satisfies}\quad  {N_j}_{\cA C}^*{N_k}_{\cA C}=\delta_{jk}P_{\cA C}\,.
 \end{gather}
In this case, $\{\no{N_j}_\varphi^{-1} N_j\}_{j=1}^t$ is an o.n.b. of $(\cA N,\varphi)$.
\end{theorem}
\begin{proof} By Lemma~\ref{lem:ip-VpN}, it is suffices to show that the inner product \eqref{eq:ip-vP-L2} is the $\cA C$-decoding inner product if and only if \eqref{eq:ip-vP-L3} holds. In this case, if \eqref{eq:ip-vP-L2} define the $\cA C$-decoding inner product, then by Theorem~\ref{teo:prop-ipN}.\eqref{c4-ipp}
\begin{align*}
{N_j}_{\cA C}^*{N_k}_{\cA C}=\frac{1}{\no{N_j}_\varphi\no{N_k}_\varphi}P_{\cA C}N_j^*N_kP_{\cA C}=\varphi(\no{N_j}_\varphi^{-1} N_j,\no{N_k}_\varphi^{-1} N_k)P_{\cA C}=\delta_{jk}P_{\cA C}\,.
\end{align*}
Conversely, \eqref{eq:ip-vP-L3} holds $P_{\cA C}N_j^*N_kP_{\cA C}=\no{N_j}_{\varphi}^2\delta_{jk}P_{\cA C}$. Thus, for $N=\sum_{j=1}^{t}\ip{N_j}{N}_2N_j$ and $M=\sum_{j=1}^{t}\ip{N_j}{M}_2N_j\in\cA N$,
\begin{align*}
P_{\cA C}N^*MP_{\cA C}&=\sum_{j=1}^{t}\sum_{k=1}^{t}\ip{N}{N_j}_2\ip{N_j}{M}_2P_{\cA C} N_j^*N_kP_{\cA C}\\&=\sum_{j=1}^{t}\no{N_j}_{\varphi}^2\ip{N}{N_j}_2\ip{N_j}{M}_2P_{\cA C}=\varphi(N,M)P_{\cA C}\,,
\end{align*}
which implies that \eqref{eq:ip-vP-L2} is the $\cA C$-decoding inner product.
\end{proof}

\section{Quantum error-correcting codes}\label{s:QECC-BIP}
In this section, we look for quantum channels that decode states supported on a code $\cA C\leq \cA H$ and corrupted by  noise $\cA N\leq B(\cA H)$.

Recall that a \emph{quantum channel} is a linear map $\Phi\colon \mathcal{B}(\cA H)\to \cA B(\cA H)$ given by 
\begin{gather}\label{eq:quantum-Channel}
\Phi\rho=\sum_{j=1}^n K_j\rho K_j^*\,,\quad\mbox{where}\quad \sum_{j=1}^n K_j^*K_j=I\,,\quad K_j\in\cA {B(H)}\,.\quad (n\in\N)
\end{gather}
The operators $K_j$ are well-known as \emph{Krauss operators}, while the structure \eqref{eq:quantum-Channel} of $\Phi$ as the \emph{Krauss decomposition}, which is not unique. For instance, we get the same quantum channel if we replace $K_j$ by $e^{i\beta} K_j$ in \eqref{eq:quantum-Channel}, with $\beta\in\R$.

\begin{definition}\label{def:QNCC}
We call a code $\cA C$ \emph{quantum $\cA N$-correcting code} ($\cA N$-cc for short) if there exists a quantum channel $\Phi$ for which
\begin{gather}\label{eq:N-cc}
\Phi\lrp{N\rho N^*}=\tr{\rho N^*N}\rho\,,\quad \mbox{for all}\quad\rho\in\tilde{\cA C}\,,\quad N\in\cA N\,.
\end{gather}
In this case, we write $(\cA C,\Phi)$ and we say that $\cA C$ is decoded by $\Phi$.
\end{definition}

\begin{remark}\label{eq:subcode-prop}
If a code $\cA C$ is an $\cA N$-cc then so is any code $\cA F\leq \cA C $, since $\tilde{\cA F}\subset \tilde{\cA C}$.
\end{remark}

By linearity, the condition \eqref{eq:N-cc} can be replaced by 
\begin{gather*}
\Phi\lrp{\sum_{m=1}^k N_m\rho N_m^*}=\tr{\rho\sum_{m=1}^k N_m^*N_m}\rho\,,\quad \mbox{for all }\rho\in\tilde{\cA C}\,,\quad N_m\in\cA N\,.\quad (k\in\N)
\end{gather*}

The \emph{negligible-noise} operators of $\cA N$ with respect to $\cA C$, is 
\begin{gather}\label{eq:negligible-noise}
\cA N_{\cA C}\colonequals \{N\in\cA N\,:\,NP_{\cA C}=0\}\leq \cA N\,.
\end{gather}
Thus, any $\cA C$ is $\cA N_{\cA C}$-cc, since for any quantum channel $\Phi$, 
\begin{gather*}
\Phi\lrp{N\rho N^*}=\Phi\lrp{NP_{\cA C}\rho N^*}=0=\tr{\rho N^*NP_{\cA C}}\rho=\tr{\rho N^*N}\rho\,,\quad \rho\in\tilde{\cA C}\,,\quad N\in\cA N_{\cA C}\,.
\end{gather*}

\begin{remark}\label{rmk:nonzero-cond} From now on we shall consider  codes $\cA C$ and noise subspaces $\cA N$ which does not contain negligible-noise operators, i.e., 
\begin{gather}\label{eq:nonzero-cond}
NP_{\cA C}\neq 0\,,\quad\mbox{for all non-zero}\quad N\in\cA N\,.
\end{gather}
Bearing in mind \eqref{eq:inner-product-S2}, the noise $\hat{\cA N}=\cA N\ominus \cA N_0$ satisfies \eqref{eq:nonzero-cond} and  $\cA N=\hat{\cA N}\oplus\cA N_0$. Thus, for $M=N+N_0\in\cA N$, with $N\in\hat{\cA N}$ and $N_0\in\cA N_0$,  
\begin{align*}
M\rho M^*=(N+N_0)\rho(N+N_0)^*=N\rho N^*\,,\quad \mbox{for } \rho\in\tilde{\cA C}\,,
\end{align*}
since $N_0\rho=N_0P_{\cA C}\rho=0$ and $\rho N_0^*=(N_0\rho)^*=0$.
\end{remark}

\begin{theorem}\label{th:KL-G}
If $(\cA C,\Phi)$ is an  $\cA N$-cc then $\cA N$ can be endowed with the $\cA C$-decoding inner product.

Conversely, if $\varphi$ is the $\cA C$-decoding inner product on $\cA N$, then for any o.n.b. $\{N_j\}_{j=1}^{n-1}$ of $(\cA N,\varphi)$ it follows that $(\cA C,\Phi)$ is an $\cA N$-cc, where
\begin{gather*}
\Phi(\rho)=\sum_{j=1}^nK_j\rho K_j^*\,,\quad K_j=P_{\cA C}N_j^*\,,\mbox{ for j=1,\dots, n-1}\quad\mbox{and}\quad K_n=I-\sum_{j=1}^{n-1}K_j^*K_j\,,
\end{gather*}
holds that $P_j=K_j^*K_j$ are orthogonal projections, for $j=1,\dots,n-1$, with $P_jP_r=\delta_{jr}P_j$, while $K_nN=0$, for all $N\in\cA N$.
\end{theorem}
\begin{proof}
If $(\cA C,\Phi)$ is an $\cA N$-cc, then for $u\neq0$ in $\cA C$ and $v\in\{u\}^\perp$, one has by \eqref{eq:N-cc} that $0=\ip{v}{\Phi(N\Dn uuN^*)v}=\sum_{j=1}^n\abs{\ip{v}{K_jNu}}^2$, whence $\ip v{K_jNu}=0$, for all $v\in\{u\}^\perp$, viz. $K_jNu=\lambda_j(N,u)u$, with $\lambda_j(N,u)\in\C$. Since $\cA C$ is a subspace, one can show that $\lambda_j(N,u)$ does not depend on $u$. So, it is of the form $\lambda_j(N)$ for some linear functional $\lambda_j\colon\cA N\to \C$, and 
\begin{gather}\label{eq:functional-condition}
K_jNu=\lambda_j(N)u\,,\quad\mbox{for all}\quad N\in\cA N\,,\quad u\in\cA C\,,\quad j=1,\dots,n\,.
\end{gather}
Define the non-negative sesquilinear form $\varphi\colon\cA N\times \cA N\to \C$ by $\varphi(N,M)=\sum_{j=1}^n\cc{\lambda_j(N)}\lambda_j(M)$. Thereby, for $u,v\in\cA C$, it follows by \eqref{eq:functional-condition} that
\begin{gather}\label{eq:iPVp-condition}
\varphi(N,M)\ip uv=\sum_{j=1}^n\ip{\lambda_j(N)u}{\lambda_j(M)v}=\sum_{j=1}^n\ip{Nu}{K_j^*K_jMv}=\ip{Nu}{Mv}\,.
\end{gather}
So, \eqref{eq:nonzero-cond} and  \eqref{eq:iPVp-condition} imply for $\no u=1$ that $\varphi(N,N)=\no{Nu}^2=0$ if and only if $N=0$. Hence, $\varphi$ 
is an inner product, and the $\cA C$-decoding inner product on $\cA N$, by \eqref{eq:iPVp-condition} and Theorem~\ref{teo:prop-ipN}. 

Conversely, one has that $P_j=K_j^*K_j=N_jP_{\cA C}N_j^*$ and by Theorem~\ref{teo:prop-ipN}.\eqref{c4-ipp},
\begin{gather*}
P_jP_r=N_jP_{\cA C}N_j^*N_rP_{\cA C}N_r^*=\varphi(N_j,N_r) N_jP_{\cA C}N_r^*=\delta_{jr}P_j\,.
\end{gather*}
Besides, for $N\in\cA N$ and $\rho\in\tilde{\cA C}$, one has  $N=\sum_{j=1}^{n-1}\varphi(N_j,N)N_j$ and again Theorem~\ref{teo:prop-ipN}.\eqref{c4-ipp},
\begin{gather*}
K_nN\rho=N\rho-\sum_{j=1}^{n-1}N_jP_{\cA C}N_j^*N\rho=N\rho-\sum_{j=1}^{n-1}\varphi(N_j,N)N_j\rho=0\,,
\end{gather*}
Hence, one gets by Theorem~\ref{teo:prop-ipN}.\eqref{c3-ipp} that $\sum_{j=1}^{n-1}\abs{\varphi(N_j,N)}^2=\varphi(N,N)=\tr{\rho N^*N}$ and  
\begin{align*}
\Phi(N\rho N^*)=\sum_{j=1}^{n-1}P_{\cA C}N_j^*N\rho N^*N_jP_{\cA C}=\sum_{j=1}^{n-1}\abs{\varphi(N_j,N)}^2\rho=\tr{\rho N^*N}\rho\,,
\end{align*}
wherefrom  $(\cA C,\Phi)$ is an $\cA N$-cc.
\end{proof}

The following are known as the Knill-Laflamme conditions, which were first proved in \cite{MR1455854}.

\begin{corollary} The following conditions are equivalent:
\begin{enumerate}[(i)]
\item\label{it1-kl} A code $\cA C$ is $\cA N$-cc.
\item\label{it2-kl} For $N,M\in\cA N$ and an o.n.b. $\{u_j\}_{j=1}^s$ of $\cA C$, there exists a unique $\lambda_{N,M}\in\C$ such that
\begin{gather}\label{eq:NL-onb}
\ip{u_i}{N^*Mu_j}=\delta_{ij}\lambda_{N,M}\,.
\end{gather}
\item\label{it3-kl} For $N,M\in\cA N$ there exists a unique $\lambda_{N,M}\in\C$ such that
\begin{gather}\label{eq:k-lC2}
P_{\cA C} N^*MP_{\cA C}=\lambda_{N,M}P_{\cA C}\,.
\end{gather}
\end{enumerate}
\end{corollary}
\begin{proof}
\eqref{it1-kl}$\Rightarrow$\eqref{it2-kl}: By Theorem~\ref{th:KL-G}, the noise $\cA N$ can be endowed with the $\cA C$-decoding inner product $\varphi$, and  Theorem~\ref{teo:prop-ipN}.\eqref{c2-ipp} implies \eqref{eq:NL-onb}.

\eqref{it2-kl}$\Rightarrow$\eqref{it3-kl}: Since $P_{\cA C}=\sum_{j=1}^s\Dn{u_j}{u_j}$, one has that
\begin{align*}
P_{\cA C} N^*MP_{\cA C}=\sum_{j=1}^s\sum_{k=1}^s\Dn{u_j}{u_j}N^*M\Dn{u_k}{u_k}=\lambda_{N,M}\sum_{j=1}^s\Dn{u_j}{u_j}=\lambda_{N,M}P_{\cA C}\,.
\end{align*}

\eqref{it3-kl}$\Rightarrow$\eqref{it1-kl}: Denote the form $\varphi\colon\cA N \times\cA N\to\C$ by $\varphi(N,M)=\lambda_{N,M}$, which is sesquilinear as a consequence of \eqref{eq:k-lC2}. Besides, 
\begin{gather*}
\cc{\varphi(N,M)}P_{\cA C}=\lrp{P_{\cA C} N^*MP_{\cA C}}^*=P_{\cA C} M^*NP_{\cA C}=\varphi(M,N)P_{\cA C}\,,
\end{gather*}
i.e., $\varphi$ is symmetric. On the other hand, $\varphi(N,N)P_{\cA C}=(NP_{\cA C})^*NP_{\cA C}\geq 0$, viz. $\varphi\geq0$, and if $\varphi(N,N)=0$, then $(NP_{\cA C})^*NP_{\cA C}=0$. So, for any $f\in\cA H$, 
$0=\ip{f}{(NP_{\cA C})^*NP_{\cA C}f}=\no{NP_{\cA C}f}^2$, which implies $NP_{\cA C}=0$, i.e, $N=0$, by virtue of \eqref{eq:nonzero-cond}, and $\varphi$ defines an inner product on $\cA N$. Thence, for all $N,M\in\cA N,\,\eta,\rho\in\tilde{\cA C}$, since $P_{\cA C}\rho=\rho$ and $P_{\cA C}\eta=\eta$, it follows by \eqref{eq:k-lC2} that 
\begin{gather*}
\varphi(N,M)\eta\rho=\eta\lrp{\lambda_{N,M}P_{\cA C}}\rho=\eta\lrp{P_{\cA C} N^*MP_{\cA C}}\rho=\eta N^*M\rho\,.
\end{gather*}
Hence, $\varphi$ is the $\cA C$-decoding inner product on $\cA N$ and Theorem~\ref{th:KL-G} implies that $\cA C$ is an $\cA N$-cc.
\end{proof}
 
Let us exhibit necessary and sufficient conditions which guarantee a code is an $\cA N$-cc.

\begin{theorem}\label{th:N-cc-Conditions}
A code $\cA C$ is an $\cA N$-cc if and only if there exists a basis  $\{N_j\}_{j=1}^{n-1}$ of $\cA N$ such that 
\begin{gather}\label{eq:basis-N-OP}
\mbox{every $N_jP_{\cA C}$ is a partial isometry onto $\cA C$\quad and\quad }\ran N_jP_{\cA C}\perp \ran N_rP_{\cA C}\,,\quad j\neq r\,.
\end{gather}
In such a case:
\begin{enumerate}[(i)]
\item\label{it:OP0} $\cA C\leq\supp  N_j$,  for $j=1,\dots,n-1$.
\item\label{it:OP1} $\Phi(\rho)=\sum_{j=1}^nK_j\rho K_j^*$ decodes $\cA C$, where its Krauss operators are given by
\begin{gather*}
K_j=P_{\cA C}N_j^*\,,\quad j=1,\dots, n-1\,,\quad \mbox{which satisfy \quad $K_j^*K_jK_r^*K_r=\delta_{jr} K_j^*K_j$ }\,,
\end{gather*}
while $K_n=I-\sum_{r=1}^{n-1}K_r^*K_r$ 
holds $K_nN_jP_{\cA C}=0$.
\item\label{it:OP2} The $\cA C$-decoding inner product on $\cA N$ is given by 
\begin{gather}\label{eq:ip-C-cc-N}
\varphi(N,M)=\sum_{j=1}^{n-1}\cc n_jm_j\,,\quad N=\sum_{j=1}^{n-1}n_jN_j,
M=\sum_{j=1}^{n-1}m_jN_j\in \cA N\,.\quad (n_j,m_j\in\C)
\end{gather}
\item\label{it:OP3} $\dim\cA N\cdot\dim \cA C\leq\dim \cA H$.
\end{enumerate}
\end{theorem}
\begin{proof}
If $\cA C$ is an $\cA N$-cc then by Theorem~\ref{th:KL-G} the noise $\cA N$ is endowed with the $\cA C$-decoding inner product $\varphi$, and any o.n.b $\{N_j\}_{j=1}^{n-1}$ of $(\cA N,\varphi)$ satisfies item~\eqref{it:OP1}. Besides, \eqref{eq:basis-N-OP}, items~\eqref{it:OP0} and \eqref{it:OP3} follow from Corollary~\ref{cor:dim-prop}.\eqref{c1:CcciP}-\eqref{c3:CcciP}.

Conversely, if a basis $\{N_j\}_{j=1}^{n-1}$ of $\cA N$ satisfies \eqref{eq:basis-N-OP}, then one has that $(N_jP_{\cA C})^*N_jP_{\cA C}=P_{\cA C}$ and 
 $\ran N_rP_{\cA C}\leq \ker (N_jP_{\cA C})^*$, for $j\neq r$, i.e.,
\begin{gather}\label{eq:comp-Nj}
P_{\cA C}N_j^*N_rP_{\cA C}=(N_jP_{\cA C})^*N_rP_{\cA C}=\delta_{jr}P_{\cA C}\,.
\end{gather}
It is clear that \eqref{eq:ip-C-cc-N} is an inner product on $\cA N$ and \eqref{eq:comp-Nj} fulfills for $u,v\in\cA C$ that
\begin{gather*}
\ip{Nu}{Mv}=\sum_{j=1}^{n-1}\sum_{r=1}^{n-1}\cc n_j m_r\ip{u}{P_{\cA C}N_j^*N_rP_{\cA C}v}=\varphi(N,M)\ip uv\,,
\end{gather*}
i.e., $\varphi$ is the $\cA C$-decoding inner product, by Theorem~\ref{teo:prop-ipN}.\eqref{c2-ipp}. If $K_j=P_{\cA C}N_j^*$ then $K_j^*K_j=N_jP_{\cA C}N_j^*$, for $j=1,\dots,n-1$, and by \eqref{eq:comp-Nj}, $
K_j^*K_jK_r^*K_r=N_jP_{\cA C}N_j^*N_rP_{\cA C}N_r^*=\delta_{jr} K_j^*K_j$. So, 
\begin{gather}\label{eq:Kn-Nj0}
K_nN_jP_{\cA C}=N_jP_{\cA C}-\sum_{r=1}^{n-1}K_r^*K_rN_jP_{\cA C}=N_jP_{\cA C}-\sum_{r=1}^{n-1}N_rP_{\cA C}N_r^*N_jP_{\cA C}=0\,.
\end{gather}
Thence, $\Phi$ as item~\eqref{it:OP1} is a quantum channel, and for $\rho\in\tilde{\cA C}$, it follows that $\rho P_{\cA C}=P_{\cA C}\rho=\rho$, while \eqref{eq:comp-Nj} and \eqref{eq:Kn-Nj0} yield
\begin{gather}\label{eq:on-Basis-Nj}
\Phi(N_r\rho N_j^*)=\sum_{s=1}^{n-1}P_{\cA C}N_s^*N_r\rho N_j ^*N_sP_{\cA C}=\delta_{jr}\rho\,.
\end{gather}
Hence, by Theorem~\ref{teo:prop-ipN}.\eqref{c3-ipp} and \eqref{eq:on-Basis-Nj}, one hast for all $\rho\in\tilde{\rho}$ and $N=\sum_{j=1}^{n-1}n_jN_j\in \cA N$ that
\begin{gather*}
\Phi(N\rho N^*)=\sum_{j=1}^{n-1}\sum_{r=1}^{n-1}\cc n_j n_r\Phi(N_r\rho N_j^*)=\varphi(N,N)\rho=\tr{\rho N^*N}\rho\,,
\end{gather*}
whence $\Phi$ decodes $\cA C$ and $(\cA C,\Phi)$ is an $\cA N$-cc
.\end{proof}

As a consequence of the above, the following assertion presents two particular limiting cases of quantum error-correcting codes.

\begin{corollary}\label{cor:C-cc-existence}
\begin{enumerate}[(i)]
\item\label{c1:C-cc} $\cA H$ is an $\cA N$-cc if and only if $\cA N=\Span\{N\}$, for some partial isometry operator $N\in\cA B(\cA H)$ onto $\cA H$. In this case, $\cA C$ is decoded by $\Phi\rho=N^*\rho N$.
\item\label{c2:C-cc} $\cA C$ is a $\cA B(\cA H)$-cc, with decoding channel $\Phi=I$ if  
$\cA C=\cA H$ and 
$\dim \cA H=1$. Otherwise, there is no a $\cA B(\cA H)$-cc. 
\end{enumerate}
\end{corollary}
\begin{proof} \eqref{c1:C-cc}: If $\cA H$ is an $\cA N$-cc, then by Theorem~\ref{th:N-cc-Conditions}.\eqref{it:OP3} it follows that $\dim\cA N=1$; that is , $\cA N=\Span\{N\}$, for some non-zero operator $N\in\cA B(\cA H)$, which by Theorems~\ref{th:KL-G} and \ref{teo:prop-ipN}.\eqref{c5-ipp} one can assume that $N$ is partial isometry onto $\cA H$. For the converse, consider $\Phi\rho=N^*\rho N$ which decodes $\cA H$.

\eqref{c2:C-cc}: If $\dim \cA H=1$, then $\dim\cA B(\cA H)=1$, which implies $\cA B(\cA H)=\Span\{I\}$. Hence, by item~\eqref{c1:C-cc} the pair $(\cA H,I)$ is a $\cA B(\cA H)$-cc. Now, for $\dim \cA H>1$, if $\cA C$ is a $\cA B(\cA H)$-cc, then Theorem~\ref{th:N-cc-Conditions}.\eqref{it:OP3} yields $1\leq \dim \cA C\leq \dim \cA H/\dim \cA B(\cA H)<1$, a contradiction.
\end{proof}
 
Taking into account Corollary~\ref{cor:C-cc-existence}.\eqref{c2:C-cc}, it follows that $\C$ is a $\cA B(\C)$-cc, with decoder $\Phi=I$.
 
\section{Code-decoding partial isometry error bases}\label{sec:CDPINB}

We have seen that for a given code $\cA C\leq \cA H$ and a noise subspace $\cA N\leq \cA B(\cA H)$, a necessary and sufficient condition so that  $\cA C$ is an $\cA N$-cc is the existence of a basis $\{N_j\}_{j=1}^{n-1}$ of $\cA N$ satisfying, for $j=1,\dots,n-1$,
\begin{align}\label{eq:SNC-Ncc}
\mbox{$N_jP_{\cA C}$ is a partial isometry onto $\cA C$\,,\quad with \quad }N_j{\cA C}\perp N_r{\cA C}\,,\quad j\neq r\,.
\end{align}

\begin{definition}\label{def:C-dnb-N}
A basis of the noise subspace $\cA N$ satisfying \eqref{eq:SNC-Ncc} is called \emph{$\cA C$-decoding partial isometry error basis} (or simply $\cA C$-decoding noise basis).
\end{definition}

The main advantage of considering a $\cA C$-decoding noise basis is that it allows us to give  a seemingly nice representation of the decoding channel  (see Theorem~\ref{th:N-cc-Conditions}.\eqref{it:OP1}) given by $\Phi(\rho)=\sum_{j=1}^nK_j\rho K_j^*$ with Krauss operators 
\begin{gather*}
K_j=P_{\cA C}N_j^*\,,\quad j=1,\dots,n-1\,,\quad\mbox{while}\quad K_n=I-\sum_{r=1}^{n-1}K_r^*K_r\,.
\end{gather*}

The above reasoning motivates us to describe the elements of code-decoding noise bases, which involves the characterization of partial isometry operators in $\cA B(\cA H)$.

\subsection{Partial isometries and the von Neumann-Wold decomposition}\label{eq:vNW-Deco}

Recall that $V\in\cA B(\cA H)$ is a partial isometry if $V^*V$ is the projection onto $\supp V$. Besides, $V$ is called:
\begin{itemize}
\item \emph{partial unitary} if $\supp V= \ran V$, and 
\item \emph{unitary} if $\supp V=\cA H$, or equivalently  if $V$ is invertible and $V^*=V^{-1}$. 
\end{itemize}
Thus, any unitary operator is partial unitary, while any partial unitary $V$ is unitary on $\supp V$.

\begin{remark}\label{rm:isometry-consequences}
\begin{enumerate}[(i)]
\item It is simple to prove the equivalence of the following statements:
\begin{enumerate}
\item $V$ is partial isometry .
\item $\no{Vf}=\no f$, for all $f\in\supp V$.
\item $\ip{Vf}{Vg}=\ip{f}{g}$, for all $f,g\in\supp V$.
\end{enumerate}
\item $V$ and $V^*$ are partial isometries simultaneously.
\item\label{it:icqs} If $V$ is a partial isometry then $VP_{\cA F}$ is a partial isometry onto $\cA F\leq \supp V$, since 
\begin{gather*}
(VP_{\cA F})^*VP_{\cA F}=P_{\cA F}V^*VP_{\cA F}=P_{\cA F}P_{\supp V}P_{\cA F}=P_{\cA F}\,,
\end{gather*}
where $P_{\supp V}$ and $ P_{\cA F}$ are the orthogonal projections onto $\supp V$ and $\cA F$, respectively. 
\end{enumerate}
\end{remark}

\begin{proposition}\label{prop:composition-VL}
If $V$ and $L$ are partial isometries then so is $VL$.
\end{proposition}
\begin{proof}
For $f\in\supp VL=\{f\in\supp L\,:\, Lf\in \supp V\}$, one has $\no{VLf}=\no{Lf}=\no f$.
\end{proof}

Given a partial isometry operator $V\in\cA B(\cA H)$, we say that the pair $(\cA L,m)$ is a \emph{wandering space} for $V$  if $\cA L\leq \cA H$,
\begin{gather}\label{eq:orthogonal-Lcond}
m=\max\{n\geq0\,:\, V^n\cA L\neq\{0\}\}\qquad\mbox{and}\qquad V^n\cA L\perp \cA L\,,\quad \mbox{for all }n\in\N\,,
\end{gather}
where $m$ is attained since \eqref{eq:orthogonal-Lcond} implies $V^n\cA L=\{0\}$, for $n\geq \dim \cA H$. Besides, the right-hand side of \eqref{eq:orthogonal-Lcond} implies 
\begin{gather}\label{eq:orth-Ws}
V^n\cA L\perp V^j\cA L\,,\quad \mbox{for $j,n\geq0$, with $n\neq j$}\,.
\end{gather}

\begin{definition}
A partial isometry operator $V\in\cA B(\cA H)$ is called \emph{unilateral-shift} if there exists a wandering space $(\cA L,m)$ for $V$, such that 
\begin{gather}\label{eq:vs-OnH}
\cA L\oplus V\cA L\oplus\dots\oplus V^m\cA L=\cA H\,.
\end{gather}
In such a case $\cA L$ is uniquely determined by $\cA L=\cA H\ominus\ran V$. Besides, \eqref{eq:vs-OnH} implies
\begin{gather}\label{eq:H-zero}
V^j\cA H=\begin{cases}
V^j\cA L\oplus\dots\oplus V^m\cA L\,,&j=0,\dots,m\\
\{0\}\,,&j>m
\end{cases}\,.
\end{gather}
\end{definition}

\begin{remark}\label{rm:trivial-shift} For a unilateral-shift $V$ with wandering space $(\cA L,m)$, it is a simple matter to verify the equivalences of following statements.
\begin{enumerate}[(i)]
\begin{multicols}{3}
\item $V=0$.
\item $\cA L=\cA H$.
\item $m=0$.
\end{multicols}
\end{enumerate}
\end{remark}

Recall that $\cA K\leq \cA H$ is invariant for $T\in\cA B(\cA H)$ (or $T$-invariant) if $T\cA K\leq \cA K$. Moreover, $\cA K$ is said to \emph{reduce} $T$ if $\cA K$ and $\cA K^\perp$ are $T$-invariant. For example, the trivial subspaces $\{0\}$ and $\cA H$ reduce every operator in $\cA B(\cA H)$.

\begin{theorem}\label{th:characterization-PIO}
For a partial isometry operator $V\in\cA B(\cA H)$ there exists a unique reducing subspace $\cA K$ for $V$, such that $V$ on $\cA K$ is unilateral-shift, while $V$ on $\cA K^\perp$ is unitary. Namely, 
\begin{gather}\label{eq:WS-decomposition}
\cA L=\cA H\ominus\ran V\qquad\mbox{and}\qquad m=\max\{n\geq0\,:\, V^n\cA L\neq\{0\}\} 
\end{gather}
yield a wandering space $(\cA L,m)$ for $V$ on  
\begin{gather*}
\cA K=\bigoplus_{j=0}^mV^j\cA L\qquad \mbox{and}\qquad \cA K^\perp=\ran V^{m+1}\,.
\end{gather*}
The space $\cA K$ may be absent or the whole space.
\end{theorem}
\begin{proof}
The assertion is simple when $V$ is unitary, whence $\cA K=\{0\}$. If $V$ is non-unitary then $(\cA L,m)$,  with $\cA L=\cA H\ominus V\cA H\neq\{0\}$, is wandering for $V$. Indeed, $\cA L\perp \ran V$ and $\cA L,V\cA H\leq\cA H$ imply
\begin{gather*}
V^n\cA L\leq V^n\cA H\leq \ran V\quad \mbox{and}\quad V^n\cA L\perp \cA L\,,\qquad n\in\N\,.
\end{gather*}
 Let $\cA K=\bigoplus_{j=0}^mV^j\cA L$ and since $V^j\cA L=V^j\cA H\ominus V^{j+1}\cA H$, for $j=0,\dots, m$,
 \begin{align*}
 \cA K=\cA L\oplus\dots\oplus  V^m\cA L=(\cA H\ominus V\cA H)\oplus\dots\oplus(V^m\cA H\ominus V^{m+1}\cA H)=\cA H\ominus  V^{m+1}\cA H\,,
 \end{align*}
 whence $\cA K^\perp=V^{m+1}\cA H$. Now, $\{0\}=V^{m+1}\cA L=V^{m+1}\cA H\ominus V^{m+2}\cA H$, viz. $V^{m+1}\cA H=V^{m+2}\cA H$, which implies  $V(\cA K^\perp)=\cA K^\perp$, i.e., $\cA K^\perp$ is $V$-invariant and $V$ is unitary on $\cA K^\perp$. On the other hand, $V\cA K= \bigoplus_{j=1}^mV^j\cA L\leq \cA K$, i.e., $\cA K$ is $V$-invariant and $V$ is clearly a unilateral-shift on $\cA K$.
 
 To prove the uniqueness, if $\cA F=\bigoplus_{j=0}^{m'}V^j\cA L'$ reduces $V$, where $(\cA L',m')$ is wandering for $V$, with $V(\cA F^\perp)=\cA F^\perp$, then 
 \begin{align*}
 \cA L&=\cA H\ominus V\cA H=(\cA F\oplus \cA F^\perp)\ominus V(\cA F\oplus \cA F^\perp)=(\cA F\oplus \cA F^\perp)\ominus (V\cA F\oplus \cA F^\perp)\\
&=\cA F\ominus V\cA F=\lrp{\bigoplus_{j=0}^{m'}V^j\cA L'}\ominus\lrp{\bigoplus_{j=1}^{m'}V^j\cA L'}=\cA L'\,.
 \end{align*}
 Hence, $\cA F=\bigoplus_{j=0}^{m'}V^j\cA L'=\bigoplus_{j=0}^{m'}V^j\cA L=\cA K$.
\end{proof}

A partial isometry operator $V\in\cA B(\cA H)$ is said to be \emph{completely non-unitary}, whenever there is no non-zero reducing subspace $\cA K$ for $V$, in which $V$ is unitary.

\begin{corollary}\label{cor:shifts-cun}
A partial isometry operator $V\in\cA B(\cA H)$ is a unilateral-shift if and only if it is completely non-unitary.
\end{corollary}
\begin{proof}
If $V$ is unilateral-shift, then one has a wandering space $(\cA L,m)$ for $V$. If $\cA K\leq\cA H$ reduces $V$ in which $V$ is unitary, then $V\cA K=\cA K$ and by \eqref{eq:H-zero},
\begin{gather*}
\cA K=V^{m+1}\cA K\leq V^{m+1}\cA H=\{0\}\,,
\end{gather*}
whence $V$ is completely non-unitary. The converse is straightforward from Theorem~\ref{th:characterization-PIO}.
\end{proof}
Theorem~\ref{th:characterization-PIO} and Corollary~\ref{cor:shifts-cun} give a decomposition of any partial isometry operator 
\begin{gather}
V=V_{\cA K}\oplus V_{\cA K^\perp}\,,
\end{gather} 
which is uniquely determined by its unitary and completely non-unitary parts, $V_{\cA K^\perp}$ and $V_{\cA K}$, respectively.  
This decomposition is well-known as the \emph{von Neumann-Wold} decomposition for isometries on infinite-dimensional Hilbert spaces \cite[Chap. I, Th. 1.1]{MR2760647}.

\subsection{$\cA C$-decoding noises bases by powers of partial isometries}\label{ss:Code-bypowers-isom}
For convenience, we shall assume non-zero partial isometries.

\begin{remark} Any partial isometry $V$ is a $\supp V$-decoding noise basis of $\Span\{V\}$. Thus, $(\supp V, \Phi)$ is a $\Span\{V\}$-cc, with
\begin{gather*}
\Phi(\rho)=P_{\ker V}\rho P_{\ker V}+V^*\rho V\,,\qquad \rho\in\cA B(\cA H)\,.
\end{gather*}
\end{remark}

By virtue of Remark~\ref{rm:trivial-shift}, Theorem~\ref{th:characterization-PIO} and Corollary~\ref{cor:shifts-cun}, for a partial isometry $V$ the following are equivalent:
\begin{enumerate}[(i)]
\item The completely non-unitary part of $V$ is zero.
\item $V$ is partial unitary.
\item The wandering space $(\cA L,m)$ of $V$ (see \eqref{eq:WS-decomposition}) satisfies $\cA L=\ker V$ (or equivalent $m=0$), which is called a \emph{simple wandering} space.
\end{enumerate}

\begin{definition}
We call a code $\cA C$ \emph{the largest $\cA N$-correcting code} if every $\cA N$-cc is a subspace of $\cA C$.
\end{definition}

\begin{theorem}\label{th:codes-Vjs}
Let $V$ be a unilateral-shift with wandering space $(\cA L,m)$ and $0\leq t\leq m$. Then, $\cA C_t=\cA L\ominus\ker V^{t}\neq\{0\}$ and 
\begin{gather*}
\{V^j\}_{j=0}^t\quad \mbox{is a $\cA C_t$-decoding noise basis of}\quad \cA N_t=\Span\{V^j\}_{j=0}^t\,.
\end{gather*}
Besides, if $\cA F_t$ is an $\cA N_t$-cc then 
\begin{gather}\label{eq:codes-Vjs}
\cA F_t\leq \lrp{\bigoplus_{j=0}^{m-t}V^j\cA L}\ominus \ker V^{t}\,.
\end{gather}
Hence, for $t=m$, it follows that  $\cA L\ominus\ker V^{m}$ is the largest $\cA N_m$-cc.
\end{theorem}
\begin{proof}
If $\cA L\ominus\ker V^{t}=\{0\}$ then $\cA L\leq \ker V^{t}$ and $V^m\cA L=\{0\}$, which by \eqref{eq:orthogonal-Lcond} is not possible. Note by \eqref{eq:vs-OnH} that $\{V^j\}_{j=0}^t$ is a basis of $\cA N$ and $\supp V^{j+1}\leq \supp V^{j}$, since $V^{*j+1}\cA H\leq V^{*j}\cA H\leq \cA H$. Thus,  
\begin{gather*}
\cA C_t=\cA L\ominus\ker V^{t}\leq\cA H\ominus\ker V^{t}\,,\quad\mbox{i.e,}\quad \cA C_t\leq \supp V^t\leq \supp V^{j}\,, 
\end{gather*}
while Proposition~\ref{prop:composition-VL} and Remark~\ref{rm:isometry-consequences}.\eqref{it:icqs} imply that $V^jP_{\cA C}$ is partial isometry onto $\cA C_t$, for $j=0,\dots,t$. Besides, $\cA C_t\leq\cA L$ and \eqref{eq:orth-Ws} yield $V^j\cA C_t\perp V^r\cA C_t$, with $j\neq r$. Hence, it follows by \eqref{eq:SNC-Ncc} that $\{V^j\}_{j=0}^t$ is a $\cA C_t$-decoding noise basis of  $\cA N_t$  and $\cA C_t$ is an $\cA N_t$-cc.

Now, if a code $\cA F_t$ is an $\cA N_t$-cc. Then, Theorem~\ref{th:KL-G} and 
Corollary~\ref{cor:dim-prop}.\eqref{c1:CcciP} imply 
\begin{gather}\label{eq:aux1-cond}
\cA F_t\leq \supp V^t=(\ker V^t)^\perp\,.
\end{gather}
Besides, for $f\in\cA F_t\leq \cA H$ it follows by \eqref{eq:vs-OnH} that $f=g_0+\dots+g_m$, with $g_j\in V^j\cA L$, for $j=0,\dots,m$, whence $V^t(g_0+\dots+g_{m-t})=V^mf$ and $g_0+\dots+g_{m-t}\in\supp V^t$. In this fashion, $f-(g_0+\dots+g_{m-t})\in\supp V^m$ and 
\begin{gather*}
V^t\lrp{f-(g_0+\dots+g_{m-t})}=V^t(g_{m-t+1}+\dots+g_m)=0\,,
\end{gather*}
whence $f=g_0+\dots+g_{m-t}\in  \oplus_{j=0}^{m-t}V^j\cA L$ and $\cA F_t\leq  \oplus_{j=0}^{m-t}V^j\cA L$. Therefore, by \eqref{eq:aux1-cond} one arrives at \eqref{eq:codes-Vjs}.\end{proof}

The following assertion is straightforward from Theorems~\ref{th:characterization-PIO} and~\ref{th:codes-Vjs}.
\begin{corollary}\label{cor:decomposition-pU-Code}
Let $V$ be a partial isometry with non-simple wandering space $(\cA L,m)$, $0\leq t\leq m$ and $\cA N_t=\Span\{V^j\}_{j=0}^t$. If $\cA F_t$ is an $\cA N_t$-cc then
\begin{gather}\label{eq:big-code-subspace}
\cA F_t \leq \ran{V^{m+1}}\oplus \lrp{\bigoplus_{j=0}^{m-t}V^j\cA L}\ominus \ker V^{t}\,.
\end{gather}
\end{corollary}

Under the conditions of Corollary~\ref{cor:decomposition-pU-Code}, any $\cA N_m$-cc belongs to $\ran V^{m+1}\oplus\cA L\ominus \ker V^{m}$. Besides, if the unitary part of $V$ is a projection (in particular, if $V$ is unilateral-shift), then 
the orthogonality condition of \eqref{eq:SNC-Ncc} and Theorem~\ref{th:codes-Vjs} imply that $\cA C=\cA L\ominus \ker V^{m}$ is the largest $\cA N_m$-cc.

\begin{remark}\label{rm:LC-closed-by-Osum}
If a noise $\cA N$ has the largest correcting code $\cA C$, then any orthogonal sum $\cA F_1\oplus\cA  F_2$ of two orthogonal $\cA N$-cc $\cA F_1$, $\cA F_2$, is as well an $\cA N$-cc. Indeed, the largest condition yields that $\cA F_1,\cA F_2\leq \cA C$, i.e., $F_1\oplus F_2\leq \cA C$ and  Remark~\ref{eq:subcode-prop} implies the assertion. 
\end{remark}

The existence of the largest $\cA N$-correcting code ensures that the set of $\cA N$-cc's is closed under orthogonal sum. Nevertheless, the largest $\cA N$-correcting code does not always exist.

\begin{example}[the non-existence of the largest correcting code]\label{ex:ne-lcc} For a Hilbert space with o.n.b. $\{e_j\}_{j=0}^{14}$, consider the partial unitary operator 
\begin{gather*}
V=\sum_{j\in\Z_4}\Dn{e_{j+1}}{e_j}+\sum_{k=4}^{11}\Dn{e_{k+3}}{e_k}\qquad (\Z_3=\{0,1,2,3\})\,.
\end{gather*}
From Theorem~\ref{th:characterization-PIO}, the wandering space of $V$ is $(\cA L,m)$ where $\cA L=\Span\{e_4,e_5,e_6\}$ and $m=3$. The unitary and completely non-unitary parts of $V$ are, respectively,
\begin{gather*}
V_{\cA K^\perp}=\sum_{j\in\Z_4}\Dn{e_{j+1}}{e_j}\,;\qquad V_{\cA K}=\sum_{k=4}^{11}\Dn{e_{k+3}}{e_k}\,,\quad\mbox{where}\quad \cA K=\Span\{e_4,\dots,e_{14}\}\,. 
\end{gather*}
For $\cA C=\Span\{e_0,e_4,e_5\}$ and $\cA F=\Span\{e_1,e_4,e_8\}$, one readily checks that $\{V^j\}_{j=0}^t$ is $\cA C$-decoding noise basis of $\cA N_t=\Span\{V^j\}_{j=0}^t$,  for $t=0,\dots,3$, while $\{V^j\}_{j=0}^2$ is $\cA F$-decoding noise basis of $\cA N_2=\Span\{V^j\}_{j=0}^2$, viz. both $\cA C$ and $\cA F$ are $\cA N_2$-cc. 

Thereby, one has by Remark~\ref{eq:subcode-prop} that $a=\Span\{e_0,e_5\}\leq\cA C$ and $b=\Span\{e_1,e_8\}\leq\cA F$ are  $\cA N_2$-cc. If $\cA N_2$ has the largest correcting code, then by Remark~\ref{rm:LC-closed-by-Osum}, $a\oplus b$ is an $\cA N_2$-cc, but  $(Va\oplus b)\cap(a\oplus b)=b$, a contradiction with the orthogonality condition of \eqref{eq:SNC-Ncc}.
\end{example}

\subsection{Cyclic unitary operators: the shift and clock operators}\label{s:shift} 
Let $U\in\cA B(\cA H)$ be a unitary operator such that there exists one-dimensional subspace $\cA L\leq \cA H$, for which
\begin{gather*}
\cA L\oplus U\cA L\oplus\dots\oplus U^m\cA L=\cA H\,,\qquad m=\dim \cA H\in\N\,.
\end{gather*}
Thus, for the group $(\Z_m,+)$ and a unit element $u\in\cA L$, it follows that $\Span\{U^ju\}_{j\in\Z_m}=\cA H$, in this case $U$ is said to be a \emph{cyclic} operator. Besides, one readily checks that $\{e_j\colonequals U^ju\}_{j\in\Z_m}$ is an o.n.b. of $\cA H$ and
\begin{gather}\label{eq:U-unitary-CO}
U=\sum_{j\in\Z_m}\Dn{e_{j+1}}{e_j}\,.
\end{gather}
The operator \eqref{eq:U-unitary-CO} is called \emph{canonical shift} (or simply shift) and gives rise to the \emph{circulant} operators
\begin{gather*}
\sum_{j\in\Z_m}c_jU^j\,,\qquad c_j\in\C\,.
\end{gather*}

Now, we regard the unitary \emph{discrete Fourier operator} 
\begin{gather*}
F\colonequals \frac{1}{\sqrt{m}}\sum_{r,j\in\Z}\zeta^{rj}\Dn{e_r}{e_j}\,,\qquad\mbox{where $\zeta\colonequals e^{2\pi i/m}$\quad satisfies\quad $\sum_{j\in\Z_m}\zeta^{jr}=m\delta_{r0}$}\,.
\end{gather*}
Besides, the \emph{entangled basis} $\{\varphi_{e_j}\}_{j\in\Z_m}$ is o.n.b. of $\cA H$, where
\begin{gather}\label{eq:e.o.n.b-vp}
\varphi_{e_j}\colonequals F^*e_j=\frac{1}{\sqrt{m}}\sum_{r\in\Z_m}\zeta^{-jr}e_r\,.
\end{gather}
Thereby, one simply computes that
\begin{gather}\label{eq:U-represet-vP}
U=\sum_{j\in\Z_m}\zeta^{j}\Dn{\varphi_{e_j}}{\varphi_{e_j}}\,,
\end{gather}
which is said to be in its \emph{entangled clock} expression.

The set $\{U^r\}_{r\in\Z_m}$ is a basis of the noise $\cA U=\Span \{U^r\}_{r\in\Z_m}\leq \cA B(\cA H)$, and in the following we shall find necessary and sufficient conditions of a given code $\cA C\leq \cA H$ for which $\{U^r\}_{r\in\Z_m}$ is a $\cA C$-decoding noise bases of $\cA U$. For instance,
\begin{gather}\label{eq:C-one-dimensional}
\cA C=\Span\{f\}\quad \mbox{for 
some non-zero $f\in\cA H$}\,,
\end{gather}
since by Theorem~\ref{th:N-cc-Conditions}.\eqref{it:OP3}, $\dim \cA C=1$.

\begin{theorem}\label{th:eio-codes} The set $\lrb{U^r}_{r\in\Z_m}$ is a $\cA C$-decoding noise basis of $\cA U$ if and only if 
\begin{gather}\label{eq:Ur-Condition-f}
\cA C=\Span\{f\}\,,\qquad\mbox{where}\quad f=\sum_{j\in\Z_m}e^{i\theta_j}\varphi_{e_j}\,,\quad \theta_j\in\R\,.
\end{gather}
\end{theorem}
\begin{proof}
It is clear that $U^rP_{\cA C}$ is partial unitary, since $U^r$ is unitary, for $j\in\Z_m$. Besides, $\cA C$ satisfies \eqref{eq:C-one-dimensional}. Thereby, by virtue of 
the orthogonality condition of \eqref{eq:SNC-Ncc} and since $U$ is unitary, it is sufficient to show that $f\perp U^sf$, for a non-zero $f\in\cA C$ and all $s\neq 0$, if and only if $f$ satisfies \eqref{eq:Ur-Condition-f}. Thus, by \eqref{eq:U-represet-vP}, since $f=\sum_{a\in\Z_m}f_a\varphi_{e_a}$, with $f_j\in\C$, one has for all $s\neq0$ that $f\perp U^sf$ if and only if 
\begin{align}\label{eq:Orthogonal-Cfj}
0=\ip{f}{U^{s}f}=\sum_{a,b,j\in\Z_m}\cc f_af_b\zeta^{j{s}}\ip{\varphi_{e_a}}{\Dn{\varphi_{e_j}}{\varphi_{e_j}}\varphi_{e_b}}=\sum_{j\in\Z_m}\abs{f_j}^2\zeta^{js}\,,
\end{align}
this if and only if $\abs{f_j}^2=\alpha>0$, for all $j\in\Z_m$. Indeed, \eqref{eq:Orthogonal-Cfj} implies the system $\sqrt{m}F \hat f=\no f^2e_0$, where $\hat f=\sum_{j\in\Z_m}\abs{f_j}^2e_j$. So, $\hat f=(m)^{-1/2}\no f^2F^*e_0=m^{-1}\no f^2\sum_{j\in\Z_m}e_j$, i.e., $\no f^2=m\abs{f_j}^2$, for all $j\in\Z_m$, whence it follows that $\abs{f_1}^2=\dots=\abs{f_m}^2>0$. The converse is straightforward. Hence, by a constant normalization one arrives at  \eqref{eq:Ur-Condition-f}.
\end{proof}

It is a simple matter to verify from \eqref{eq:e.o.n.b-vp} that 
\begin{gather}\label{eq:o.n.b-e-j}
e_j=F\varphi_{e_j}=\frac{1}{\sqrt{m}}\sum_{r\in\Z_m}\zeta^{jr}\varphi_{e_r}\,,\qquad j\in\Z_m\,.
\end{gather}

\begin{corollary}\label{cor:decoder-vpj} For $j\in\Z_m$, the code $\cA C_j=\Span\{e_j\}$ is a $\cA U$-cc, with decoder
\begin{gather}\label{eq-Dec-ej1}
\Phi(\rho)=\sum_{r\in\Z_m}\ip{e_{j+r}}{\rho e_{j+r}}P_{\cA C_j}\,,\qquad \rho\in\cA B(\cA H)\,.
\end{gather}
\end{corollary}
\begin{proof}
It follows from Theorems~\ref{th:N-cc-Conditions}, \ref{th:eio-codes} and condition \eqref{eq:o.n.b-e-j} that $\cA C_j$ is a $\cA U$-cc. Besides, the decoder \eqref{eq-Dec-ej1} follows from Theorem~\ref{th:N-cc-Conditions}.\eqref{it:OP1},  
since the Krauss operators of $\Phi$ satisfies $K_r=\Dn{e_j}{e_j}U^{*r}=\Dn{e_j}{e_{j+r}}$ and $\sum_{r\in\Z_m}K_rK_r^*=I$.
\end{proof}

Now, we consider the \emph{canonical clock} operator (or simply clock operator)
\begin{gather}\label{eq:V-unitary-CO}
V\colonequals\sum_{j\in\Z_m}\zeta^{j}\Dn{e_j}{e_j}\,,
\end{gather}
which certainty is a cyclic unitary operator, since one readily expresses  \eqref{eq:V-unitary-CO} in its \emph{entangled shift} version
\begin{gather*}
V= \sum_{j\in\Z_m}\Dn{\varphi_{e_{j-1}}}{\varphi_{e_j}}\,,\quad\mbox{and}\quad \bigoplus_{j\in\Z_m} V^j\lrp{\Span\{\varphi_{e_0}\}}=\cA H\,.
\end{gather*}
In this fashion, bering in mind \eqref{eq:U-represet-vP}, one can change the roll of the basis $\{e_j\}_{j\in\Z_m}$ and $\{\varphi_{e_j}\}_{j\in\Z_m}$ in Theorem~\ref{th:eio-codes} and Corollary~\ref{cor:decoder-vpj} to get the following result.

\begin{theorem}\label{th:cio-codes}
The set $\lrb{V^r}_{r\in\Z_m}$ is a $\cA C$-decoding noise basis of $\cA V=\Span\lrb{V^r}_{r\in\Z_m}$ if and only if 
\begin{gather*}
\cA C=\Span\{g\}\,,\qquad\mbox{where}\quad g=\sum_{j\in\Z_m}e^{i\theta_j}e_j\,,\quad \theta_j\in\R\,.
\end{gather*}
In particular, for $j\in\Z_m$, the code $\cA C_j=\Span\{\varphi_{e_j}\}$ is a $\cA V$-cc, with decoder
\begin{gather*}\label{eq-Dec-ej}
\Phi(\rho)=\sum_{r\in\Z_m}\ip{\varphi_{e_{j-r}}}{\rho\varphi_{e_{j-r}}}P_{\cA C_j}\,,\qquad \rho\in\cA B(\cA H)\,.
\end{gather*}
\end{theorem}

We shall tackle the shift and clock operators in the section of reducing subspaces \ref{sec:ccn}.

\section{Code-decoding error basis in reducing subspaces}\label{sec:ccn}

In what follows, we turn the attention to code-decoding noise basis in reducing subspaces. 
\begin{remark}\label{rm:sub-Cdnb}
One readily checks that if $\{N_j\}_{j=1}^{n}$ is a $\cA C$-decoding noise basis of $\cA N$, then it follows for $m\leq n$ that $\{N_j\}_{j=1}^{m}$ is $\cA C$-decoding noise of $\Span\{N_j\}_{j=1}^{m}$. 
\end{remark}

For a partition $\cA H_1,\dots,\cA H_k$ of $\cA H$, if $\{N_{sj}\}_{j=1}^{n_s}$ is a $\cA C_s$-decoding noise basis of $\cA N_s\leq \cA B(\cA H_s)$, then by Remark~\ref{rm:sub-Cdnb}, one has that $\{N_{sj}\}_{j=0}^{n}$ is a $\cA C_s$-decoding noise basis of $\Span\{N_{sj}\}_{j=0}^{n}$, where $n=\min\{n_s\,:\, s=1,\dots, k\}$. 

\begin{proposition}\label{prop:C-db-onHj1} Let $n,k\in\N$ and $\cA H_1,\dots,\cA H_k$ be a partition of $\cA H$. Then, $\{N_{sj}\}_{j=1}^{n}$ is a $\cA C_s$-decoding noise basis of $\cA N_s\leq \cA B(\cA H_s)$, for $s=1,\dots,k$, if and only if  
\begin{gather}\label{eq:C-db-onHj}
\lrb{\bigoplus_{s=1}^kN_{sj}}_{j=1}^{n}\quad\mbox{is a}\quad \bigoplus_{s=1}^{k}\cA C_s\mbox{-decoding noise basis of}\quad \bigoplus_{s=1}^k\cA N_{s}\,.
\end{gather}
\end{proposition}
\begin{proof} 
If $\{N_{sj}\}_{j=1}^{n}$ is a $\cA C_s$-decoding noise basis of $\cA N_s$, for $s=1,\dots,k$,  then it follows by the orthogonality of the Hilbert spaces that $\bigoplus_{s=1}^{k}\cA C_s\leq \supp \bigoplus_{s=1}^kN_{sj}$ and for $j=1,\dots,n$,
\begin{gather*}
\lrp{\bigoplus_{r=1}^kN_{rj}P_{\bigoplus_{s=1}^{k}\cA C_s}}^*\bigoplus_{r=1}^kN_{rj}P_{\bigoplus_{s=1}^{k}\cA C_s}=P_{\bigoplus_{s=1}^{k}\cA C_s}\bigoplus_{r=1}^kN_{rj}^*N_{rj}P_{\bigoplus_{s=1}^{k}\cA C_s}=P_{\bigoplus_{s=1}^{k}\cA C_s}\,.
\end{gather*}
Besides, $\bigoplus_{r=1}^kN_{rj}\bigoplus_{s=1}^{k}\cA C_s=\bigoplus_{s=1}^kN_{sj}\cA C_s$ implies $\bigoplus_{r=1}^kN_{rj}\bigoplus_{s=1}^{k}\cA C_s\perp \bigoplus_{r=1}^kN_{rt}\bigoplus_{s=1}^{k}\cA C_s$, since $N_{sj}\cA C_s\perp N_{st}\cA C_s $, for $j\neq t$, which implies \eqref{eq:C-db-onHj}. Conversely, if  \eqref{eq:C-db-onHj} is true then Remarks~\ref{eq:subcode-prop} and~\ref{rm:sub-Cdnb} implies that  $\{N_{sj}\}_{j=1}^{n_s}$ is a $\cA C_s$-decoding noise basis of $\cA N_s$, for $s=1,\dots,k$.
\end{proof}

Recall that $\cA H_1,\dots,\cA H_k\leq\cA H$ are reducing subspaces for $V\in\cA B(\cA H)$, if all of them are $V$-invariant and $\bigoplus_{j=1}^k\cA H_j=\cA H$. If such is the case, $V=\bigoplus_{j=1}^kV_j$, where
\begin{gather*}
V_j\colonequals VP_{\cA H_j}=P_{\cA H_j}VP_{\cA H_j}\quad\mbox{satisfies}\quad  \ran V_j=\ran V\cap\cA H_j\quad\mbox{and}\quad\supp V_j=\supp V\cap\cA H_j\,.
\end{gather*}
The operator $V_j$ is called the \emph{reduced operator} by $\cA H_j$.

\begin{lemma}\label{lem:V-reducedHj}
For reducing subspaces $\cA H_1,\dots,\cA H_k$ for $V\in\cA B(\cA H)$, with reduced operators $V_1,\dots,V_k$, respectively, the following are true:
\begin{enumerate}[(i)]
\item\label{it1:V-rk} $V$ is partial isometry if and only every $V_j$ is partial isometry.
\item\label{it2:V-rk} $V$ is partial unitary if and only every $V_j$ is partial unitary.
\item\label{it3:V-rk} $V$ is unilateral-shift if and every $V_j$ is unilateral-shift. In such a case, the 
wandering space $(\cA L,m)$ for $V$ is 
\begin{gather*}
\cA L=\bigoplus_{j=1}^k\cA L_j\,,\qquad m=\max\{m_j\}_{j=1}^k\,,
\end{gather*}
being $(\cA L_j,m_j)$ wandering for $V_j$, $j=1,\dots,k$.

\end{enumerate}
\end{lemma}
\begin{proof}\eqref{it1:V-rk}: If
$V$ is partial isometry then, since $P_{\supp V}$ and $P_{\cA H_j}$ commute, 
\begin{gather*}
V_j^*V_j=P_{\cA H_j}V^*VP_{\cA H_j}=P_{\cA H_j}P_{\supp V}P_{\cA H_j}=P_{\supp V_j}\,,
\end{gather*}
whence $V_j$ is partial isometry. Conversely, since $V=\bigoplus_{j=1}^kV_j$ and by the orthogonality condition of the reducing subspaces 
\begin{gather*}
V^*V=\lrp{\bigoplus_{j=1}^kV_j}^*\lrp{\bigoplus_{r=1}^kV_r}=\bigoplus_{j=1}^kV_j^*V_j=\bigoplus_{j=1}^kP_{\supp V\cap \cA H_j}=\bigoplus_{j=1}^kP_{\supp V}P_{\cA H_j}=P_{\supp V}\,,
\end{gather*}
i.e., $V$ is partial isometry.

\eqref{it2:V-rk}: If $V$ is partial unitary, then $\supp V=\ran V$ and $V_j$ is partial isometry, by \eqref{it1:V-rk}. Besides, since $\cA H_j$ reduces $V$,
\begin{gather*}
\ran V_j=\ran V\cap{\cA H_j}=\supp V\cap{\cA H_j}=\supp V_j\,,
\end{gather*}
i.e., $V_j$ is partial unitary. On the other hand, if $V_j$ is partial unitary then $\supp V_j=\ran V_j$. Thus, $V$ is partial isometry by \eqref{it1:V-rk} and  
\begin{align*}
\ran V&=\bigoplus_{j=1}^k\ran V\cap{\cA H_j}=\bigoplus_{j=1}^k\ran V_j\\&=\bigoplus_{j=1}^k\supp V_j=\bigoplus_{j=1}^k\supp V\cap{\cA H_j}=\supp V\,,
\end{align*}
whence $V$ is partial unitary.

\eqref{it3:V-rk}: If $V$ is unilateral-shift then there exists a wandering space $(\cA L,m)$ for $V$ such that \eqref{eq:vs-OnH} holds. In this fashion, \eqref{it1:V-rk} implies that $V_j$ is partial isometry in $\cA H_j$ and for $\cA L_j=\cA L\cap \cA H_j$, 
\begin{gather*}
\bigoplus_{r=0}^m V_j^r\cA L_j=P_{\cA H_j}\bigoplus_{r=0}^m V^rP_{\cA H_j}\cA L_j=P_{\cA H_j}\lrp{\bigoplus_{r=0}^m V^r\cA L}\cap \cA H_j=\cA H_j\,.
\end{gather*}
Thereby, taking $m_j=\max\{n\geq0\,:\, V_j^n\cA L_j\neq\{0\}\}\leq m$, it follows that $(\cA L_j,m_j)$ is wandering for $V_j$ and, hence, $V_j$ is unilateral-shift in $\cA H_j$. Conversely, if $V_j$ is unilateral-shift in $\cA H_j$, with wandering space $(\cA L_j,m_j)$, for $j=1,\dots,k$, then  it follows by \eqref{it2:V-rk} that $V$ is partial unitary. Besides, for $m=\max\{m_j\}_{j=1}^k$ and $\cA L=\bigoplus_{j=1}^k\cA L_j$, one gets that 
\begin{gather*}
\bigoplus_{r=0}^m V^r\cA L=\bigoplus_{r=0}^m\bigoplus_{j=1}^kV_j^r\bigoplus_{s=0}^k\cA L_s=\bigoplus_{j=1}^k\bigoplus_{r=0}^mV_j^r\cA L_j=\bigoplus_{j=1}^k\cA H_j=\cA H\,.
\end{gather*}
Therefore, $(\cA L,m)$ is wandering for $V$ and $V$ is unilateral-shift.
\end{proof}

The following is straightforward from Theorem~\ref{th:characterization-PIO} and Lemma~\ref{lem:V-reducedHj}.\eqref{it3:V-rk}.

\begin{corollary}\label{cor:wNw-Deco-RS}
Let $V$ be a partial isometry, $\cA H_1,\dots,\cA H_k$ be reducing subspaces for $V$ and $V_j=VP_{\cA H_j}$, for $j=1,\dots, k$. Then, for
 \begin{gather*}
\cA L_j=\cA H_j\ominus \ran V_j\qquad\mbox{and}\qquad m_j=\max\{n\geq0\,:\, V_j^n\cA L_j\neq\{0\}\}\,,
\end{gather*}
the unique reducing subspace in which $V_j$ is unilateral-shift,  
with wandering space $(\cA L_j,m_j)$, is
\begin{gather*}
\cA K_j=\bigoplus_{r=0}^{m_j}V_j^r\cA L_j\leq\cA H_j\,,\quad\mbox{while\quad $V_j$ on $\cA H_j\ominus \cA K_j=\ran V_j^{m_j+1}$ is unitary}\,.
\end{gather*}
Therefore, the unique reducing subspace in which $V$ is unilateral-shift with wandering space $(\bigoplus_{j=1}^k\cA L_j,\max\{m_j\}_{j=1}^k)$, is 
\begin{gather*}
\cA K=\bigoplus_{j=1}^k\cA K_j\,,\qquad\mbox{\quad and\quad $V$ on}\quad \cA K^\perp=\bigoplus_{j=1}^k\ran V_j^{m_j+1}\quad\mbox{is unitary}\,.
\end{gather*}
\end{corollary}
In Corollary~\ref{cor:wNw-Deco-RS}, the decomposition of every restriction $V_j$ of $V$ into $V_j={V_j}_{\cA K_j}\oplus {V_j}_{\cA K_j^\perp}$, being $ {V_j}_{\cA K_j^\perp}$ and ${V_j}_{\cA K_j} $ the unitary and completely non-unitary parts of $V_j$, respectively, uniquely determines the decomposition of $V=V_{\cA K}\oplus V_{\cA K^\perp}$ into its unitary $V_{\cA K^\perp}=\bigoplus_{j=1}^k {V_j}_{\cA K_j^\perp}$ and completely non-unitary $V_{\cA K}=\bigoplus_{j=1}^k {V_j}_{\cA K_j}$ parts.

\begin{theorem}\label{th:rS-Pisom-Ncc}
Let $\cA H_1,\dots,\cA H_k$ be reducing subspaces for $V$, such that the reduced operator $V_j$ is unilateral-shift with wandering space $(\cA L_j,m_j)$, for $j=1,\dots,k$.  Then, for $m_j\neq0$ and $0\leq t_j\leq m_j$, it follows that 
$\cA C_{t_j}=\cA L_j\ominus \ker V_j^{t_j}\neq \{0\}$ and 
\begin{gather}\label{eq:LNcc-Vj}
\{V^r\}_{r=0}^{t_j}\quad \mbox{is a $\cA C_{t_j}$-decoding noise basis of}\quad \cA N_{t_j}=\Span\{V^r\}_{r=0}^{t_j}\,.
\end{gather}
Besides, for $m=\max\{m_j\}_{j=1}^k$, $0\leq t\leq m$ and $\cA N_t=\Span\{V^r\}_{r=0}^t$,  if $\cA F_t$ is an $\cA N_t$-cc then 
\begin{gather*}
\cA F_t\leq \bigoplus_{j=1}^k\lrp{\bigoplus_{r=0}^{m-t}V_j^r\cA L_j}\ominus \ker V_j^{t}\,.
\end{gather*}
Moreover, $\cA N_m$ has the largest correcting code given by $\bigoplus_{j=1}^k\cA L_j\ominus \ker V_j^{m}$.
\end{theorem}
\begin{proof}
The first part of the assertion follows by Remark~\ref{rm:trivial-shift} and Theorem~\ref{th:codes-Vjs}, since $V^rP_{\cA C_{t_j}}=V_j^rP_{\cA C_{t_j}}$. Now, for the second part, if $\cA F_t$ is an $\cA N_t$-cc, then by \eqref{eq:codes-Vjs} and Lemma~\ref{lem:V-reducedHj}.\eqref{it3:V-rk},
\begin{align*}
\cA F_t&\leq \lrp{\bigoplus_{r=0}^{m-t}V^r\lrp{\bigoplus_{j=1}^k\cA L_j}}\ominus \ker V^{t}=\lrp{\bigoplus_{j=1}^k\lrp{\bigoplus_{r=0}^{m-t}V_j^r\cA L_j}}\ominus\lrp{\bigoplus_{s=1}^k\ker V_s^t}\\
&=\bigoplus_{j=1}^k\lrp{\bigoplus_{r=0}^{m-t}V_j^r\cA L_j}\ominus \ker V_j^{t}\,,
\end{align*}
since $\ker V^t=\cap_{j=1}^k\ker V^t_j$, viz. $(\ker V^t)^\perp=\bigoplus_{j=1}^k(\ker V^t_j)^\perp$. To conclude, one has by Theorem~\ref{th:codes-Vjs} and Lemma~\ref{lem:V-reducedHj}.\eqref{it3:V-rk} that $\cA N_m$ has the largest correcting code given by 
\begin{gather*}
\lrp{\bigoplus_{j=1}^k\cA L_j}\ominus\ker V^{m}=\lrp{\bigoplus_{j=1}^k\cA L_j}\ominus\lrp{\bigoplus_{s=1}^k\ker V_s^m}=\bigoplus_{j=1}^k\cA L_j\ominus \ker V_j^{m}\,,
\end{gather*}
as required.
\end{proof}

The following result is straightforward from Corollary~\ref{cor:decomposition-pU-Code} and Theorem~\ref{th:rS-Pisom-Ncc}.

\begin{corollary}
Let $\cA H_1,\dots,\cA H_k$ be reducing subspaces for a partial isometry $V$, such that for $j=1,\dots,k$, some of the wandering spaces $(\cA L_j,m_j)$ of the partial isometry reduced operator $V_j$ is non-simple. For $0\leq t\leq \max\{m_j\}_{j=1}^k$ and $\cA N_t=\Span\{V^j\}_{j=0}^t$, if  $\cA F_t$ is an $\cA N_t$-cc then
\begin{gather*}
\cA F_t\leq\bigoplus_{j=1}^k\ran{V_j^{m+1}}\oplus\lrp{\bigoplus_{r=0}^{m-t}V_j^r\cA L_j}\ominus \ker V_j^{t}\,.
\end{gather*}
\end{corollary}

We now turn to the analysis of cyclic unitary operators in reducing subspaces (see subsection~\ref{s:shift}). By abuse of notation, we regard a partition $\cA H_1,\dots,\cA H_k$ of $\cA H$, with (respectively) o.n.b.'s $\{\vk{1_j}\}_{j_1\in\Z_{m_1}},\dots,\{\vk{k_j}\}_{j\in\Z_{m_k}}$, and entangled basis $\{\varphi_{1_j}\}_{j_1\in\Z_{m_1}},\dots,\{\varphi_{k_j}\}_{j\in\Z_{m_k}}$.

Consider the shift and clock operators, 
 \begin{gather}\label{eq:can-ent-SO}
 U_s\colonequals\sum_{j\in\Z_{m_s}}\Dn{s_{j+1}}{s_j}\,,\quad V_s\colonequals\sum_{j\in\Z_{m_s}}\zeta^{j}\Dn{s_j}{s_j}\,,\quad\mbox{(respectively)}\quad s=1,\dots,k\,.
 \end{gather}
If $\cA U_s=\Span\{U_s^r\}_{r\in\Z_{m_s}}$ and $\cA V_s=\Span\{V_s^r\}_{r\in\Z_{m_s}}$, then from Theorems~\ref{th:eio-codes} and~\ref{th:cio-codes}, every quantum $\cA U_s$-correcting code $\cA C_{\cA U_s}$ and quantum $\cA V_s$-correcting code $\cA C_{\cA V_s}$ admit the representation
 \begin{gather*}
 \cA C_{U_s}=\Span\lrb{\sum_{j\in\Z_{m_s}}e^{i\theta_j}\varphi_{s_j}\,:\, \theta_j\in\R}\,,\quad 
  \cA C_{V_s}=\Span\lrb{\sum_{j\in\Z_{m_s}}e^{i\theta_j}{\vk{s_j}}\,:\, \theta_j\in\R}\,.
 \end{gather*}
 
 \begin{theorem}\label{th:codes-Rs-ONh}
For $s=1,\dots,k$, consider $\cA N_{s}=\Span\{T_s^r\}_{r\in\Z_m}$, where
$T_s\in\{U_s,V_s\}$ and $m=\min\{m_s\,:\,s=1,\dots,k\}$. Then, 
\begin{gather}\label{eq:C-db-rsHj}
\lrb{\bigoplus_{s=1}^kT_s^r}_{r\in\Z_m}\quad\mbox{is a}\quad \bigoplus_{s=1}^{k}\cA C_{T_s}\mbox{-decoding noise basis of}\quad \bigoplus_{s=1}^k\cA N_{s}\,.
\end{gather}
 \end{theorem}
\begin{proof}
Since $m\leq m_s$, it follows by Remark~\ref{rm:sub-Cdnb} that $\{T_s^r\}_{r\in\Z_m}$ is $\cA C_{T_s}$-decoding noise basis of $\cA N_s$, for $s=1,\dots,k$. Hence, Proposition~\ref{prop:C-db-onHj1} yields \eqref{eq:C-db-rsHj}.
\end{proof}

The proof of the following assertion readily follows from Theorems~\ref{th:N-cc-Conditions} and~\eqref{th:codes-Rs-ONh}, taking into account the conditions~\eqref{eq:e.o.n.b-vp} and \eqref{eq:o.n.b-e-j}.

\begin{corollary} For $m=\min\{m_s\,:\,s=1,\dots,k\}$ and $j_s\in\Z_{m_s}$, the code $\cA C_j=\Span\{\vk{s_{j_s}}\}_{s=1}^k$ is a $\Span\{\bigoplus_{s=1}^kU_s^r\}_{r\in\Z_m}$-cc, with decoder
\begin{gather*}
\Phi(\rho)=\bigoplus_{s=1}^k\lrp{\sum_{r\in\Z_m}\ip{s_{j_s+r}}{\rho s_{j_s+r}}\Dn{s_{j_s}}{s_{j_s}}+\sum_{r=m}^{m_s-1}\ip{s_{j_s+r}}{\rho s_{j_s+r}}\Dn{s_{j_s+r}}{s_{j_s+r}}}\,,
\end{gather*}
while $\cA F_j=\Span\{\varphi_{s_{j_s}}\}_{s=1}^k$ is a $\Span\{\bigoplus_{s=1}^kV_s^r\}_{r\in\Z_m}$-cc, with decoder
\begin{gather*}
\Phi(\rho)=\bigoplus_{s=1}^k\lrp{\sum_{r\in\Z_m}\ip{\varphi_{s_{j_s-r}}}{\rho\varphi_{s_{j_s-r}}}\Dn{\varphi_{s_{j_s}}}{\varphi_{s_{j_s}}}+\sum_{r=m}^{m_s-1}\ip{\varphi_{s_{j_s-r}}}{\rho \varphi_{s_{j_s-r}}}\Dn{\varphi_{s_{j_s-r}}}{\varphi_{s_{j_s-r}}}}\,.
\end{gather*}
\end{corollary}

\begin{remark}\label{rm:weyl-operators}
It is worth mentioning that the shift and clock operators \eqref{eq:can-ent-SO} give rise to the so-called \emph{Weyl operators} (c.f. \cite{2006KRP}). Namely, the unitary operators $W_{(r,t)}=\bigoplus_{s=1}^kU_s^{r_s}V_s^{t_s}\in\cA B(\cA H)$, rewritten by
\begin{gather*}
W_{(r,t)}=\bigoplus_{s=1}^k\sum_{j\in\Z_{m_s}}\zeta^{jt_s}\Dn{s_{j+r_s}}{s_j}\,,\qquad r=(r_1,\dots,r_k),\,t=(t_1,\dots,t_k)\in \bigoplus_{s=1}^k\Z_{m_s}\,.
\end{gather*}
The set of all Weyl operators forms an orthogonal unitary basis of $(\cA B(\cA H),\ip\cdot\cdot_2)$ and plays an essential role in the theory of 
the stabilizer and $t$-error correcting quantum codes \cite{arab:hal-04560021,MR3051751}.
\end{remark}

\subsection*{Acknowledgment}
The first and third authors were partially supported by CONACYT-Mexico Grants CBF2023-2024-1842, CF-2019-684340, and UAM-DAI 2024: ``Enfoque Anal\'itico-Combinatorio y su Equivalencia de Estados Gaussianos". 

The second author was partially supported by SECIHTI-Mexico Grants CBF 2023-2024-224 and CF-2023-G-33.

\def\cprime{$'$} \def\lfhook#1{\setbox0=\hbox{#1}{\ooalign{\hidewidth
  \lower1.5ex\hbox{'}\hidewidth\crcr\unhbox0}}} \def\cprime{$'$}
  \def\cprime{$'$} \def\cprime{$'$} \def\cprime{$'$} \def\cprime{$'$}
  \def\cprime{$'$} \def\cprime{$'$}
\providecommand{\bysame}{\leavevmode\hbox to3em{\hrulefill}\thinspace}
\providecommand{\MR}{\relax\ifhmode\unskip\space\fi MR }
\providecommand{\MRhref}[2]{%
  \href{http://www.ams.org/mathscinet-getitem?mr=#1}{#2}
}
\providecommand{\href}[2]{#2}

\end{document}